\documentclass{styles/svproc}
\usepackage{url}
\usepackage{float}
 \usepackage{setspace}
\usepackage{amsmath}
\usepackage{graphicx}

\begin{document}
\mainmatter              
\title{The complexity of the Timetable-Based Railway Network Design Problem}
\titlerunning{The complexity of the Timetable-Based Railway Network Design Problem}  
%
\author{Nadine Friesen\inst{1} \and Tim Sander\inst{2} \and Christina Büsing \inst{3}\and Karl Nachtigall\inst{2} \and Nils Nießen \inst{1}}
\authorrunning{Nadine Friesen, Tim Sander et al.} 
%
\tocauthor{Nadine Friesen, Tim Sander, Karl Nachtigall, Nils Nießen}
\institute{Institute of Transport Science, RWTH Aachen, Aachen, Germany\\
\email{friesen@via.rwth-aachen.de},\\ WWW home page:
\texttt{https://www.via.rwth-aachen.de/}
\and
 Chair of Traffic Flow Sciences, TU Dresden,
Dresden, Germany
\and
Lehr- und Forschungsgebiet Kombinatorische Optimierung, RWTH Aachen
Aachen, Germany}

\maketitle              

\begin{abstract}
	Because of the long planning periods and their long life cycle, railway infrastructure has to be outlined long ahead. At the present, the infrastructure is designed while only little about the intended operation is known. Hence, the timetable and the operation are adjusted to the infrastructure. Since space, time and money for extension measures of railway infrastructure are limited, each modification has to be done carefully and long lasting and should be appropriate for the future unknown demand.

To take this into account, we present the robust network design problem for railway infrastructure under capacity constraints and uncertain timetables. Here, we plan the required expansion measures for an uncertain long-term timetable. We show that this problem is NP-hard even when restricted to bipartite graphs and very simple timetables and present easier solvable special cases.

This problem corresponds to the fixed-charge network design problem where the expansion costs are minimized such that the timetable is conductible. We model this problem by an integer linear program using time expanded networks.

To incorporate the uncertainty of the future timetable, we use a scenario-based approach. We define scenarios with individual departure and arrival times and optional trains. The network is then optimized such that a given percentage of the scenarios can be operated while minimizing the expansion costs and potential penalty costs for not scheduled optional trains. 

\keywords{Network Design, Timetabling, Robust Optimization, Railway Planning, Railway Network Design, Strategic Timetabling}
\end{abstract}
\section{Introduction}
Traditionally, the planning process for public transportation and railways is executed sequentially: infrastructure planning, line planning and timetable planning. Network design is included in the infrastructure while timetabling is usually executed later in the planning horizon. Because of that, the infrastructure is often assumed to be fixed during the tactical timetabling. This leads to strict constraints for the timetable and reduces the options to adjust the timetable to the transportation demand which possibly changes between the infrastructure planning and the actual operation. Therefore more and more railway companies in Europe shift from this traditional approach to a timetable-based one by including a long-term timetable into infrastructure planning. This long-term timetable is then used as input for further planning steps, including network design. This approach is used in several western European countries, including Switzerland, Germany and the Netherlands. However, it is rarely covered in the scientific literature. The ideas and the conception of strategic timetables in Germany have been described by Weigand et Heppe in \cite{heppe2019spurplangestaltung}. Only recently, Polinder et al proposed a mathematical program for strategic timetabling in \cite{polinder2021timetabling}. 

Network design problems describe the decision which nodes or lines of a network should be expanded or build to meet a given demand. They occur in many different contexts for example in transportation, communication or electricity. The classic network design problem is described in \cite{leblanc1975algorithm} by Leblanc. The problem is modelled as a nonlinear mixed-integer problem. An overview over models and algorithms for transportation network design is given in \cite{magnanti1984network}. A more recent review by Chen et al. can be found in \cite{chen2011transport}. The fixed charge capacited network design problem with multi-commodity flow is closely related to the railway network design problem. If an arc is used, a fixed charge is applied, independently of the amount of flow that is transported over this arc. A survey of this problem is given in \cite{gendron1999multicommodity}. They also provide possible relaxations of the problem. The capacited network design problem for a multi-commodity flow describes a network design problem where the capacity of each arc is limited and the flow starts from different sources and should arrive at different sinks. In \cite{costa2009benders} this problem is described and strategies to solve this problem with Benders decomposition are analyzed. 

In the context of network design for railway infrastructure in the traditional planning approach, vague demand prognoses are assumed. Based on these, the necessary capacity expansions and new constructions are determined. The network design problem in railways without a timetable or temporal constraints has been studied for example in \cite{sponemann2013network} where Sp{ö}nemann provides a MIP formulation for the problem as well as some computational results. In \cite{kuby2001railway} Kuby et al. give an example of the railway network design where extensions of the network are determined by using a mixed-integer linear program. In a second step, they define different stages of the extensions through a heuristic backwards time sequencing procedure.

Most network design problems do not consider a temporal aspect and assume a static demand that does not vary over time. This can be assumed if a worst-case demand over all time steps for the network can be given. However, if such a worst-case cannot be determined, the timetable needs to be included in the network. As the problem becomes computationally intensive if we additionally include a temporal dimension, there exist only very few approaches that consider a temporal aspect. In \cite{guihaire2008transit} Guihaire et al. give a survey over transit network design. They include the planning of the transit routes network, the frequencies and departure times which covers problem related to both network design and timetabling. Zhao et al. combine in \cite{zhao2008optimization}  route network design, vehicle headway and timetable assignment and provide a metaheuristic search scheme consisting of simulated annealing, tabu and greedy search methods to solve the problem.

Timetables in the railway network design problem have been studied in \cite{schobel2012timetable}. Schöbel at al. provide a formalization of the network design problem under timetable constraints as well as sketch some possible algorithms and heuristics. No implementations or computational results are given. In \cite{schobel2013optimization}, the approach is extend with some considerations to reduce the computation time of the problem. In \cite{grujivcic2015variable}, a neighborhood search algorithm is given where as in \cite{grujivcic2017metaheuristic} a meta-heuristic approach is developed.

Due to the long planning horizons for railway infrastructure, the strategic timetable is subject to uncertainty. One method to handle this uncertainty is robust optimization. In the context of railways, robust optimization has been extensively studied. A survey about robust optimization in railway planning can be found in \cite{lusby2018survey}. However, most research considering robustness in railways focuses on robust timetables and not on robust infrastructure. Results about robust infrastructure in the case of failure which occur in the operation can be found in \cite{garcia2013grasp}.

Robustness in network design problems can be found in different contexts: Ukkusuri et al. \cite{ukkusuri2007robust} consider the network design problem under demand uncertainties for transportation networks. The uncertainty in the demand is modelled through random variables with known distributions for the entries of the origin-destination-matrices. Pishvaee et al. (\cite{pishvaee2012robust}) describe a probabilistic programming approach for the supply network design problem under uncertainties. In \cite{chen2011transport} Chen and al. present a bi-objective-reliable network design problem model that optimizes the reliability of the capacity and the travel time under demand uncertainty. A simulation-based multi-objective genetic algorithm solution procedure is developed.
In \cite{mudchanatongsuk2008robust} Mudchanatongsuk et al. present a robust optimization-based formulation for the network design problem under transportation cost and demand uncertainty. An approximation to this is shown to be done efficiently for a network with single origin and destination per commodity and general uncertainty in transportation costs and demand that are independent of each other. For a network with path constraints, an efficient column generation procedure to solve the linear programming relaxation is given.
These approaches do not take a temporal aspect into account.

Up to the authors' knowledge, little research considering the network design problem under both uncertain demand and timetable constraints has been conducted yet. In \cite{sander2023robust}, we integrate the timetable conditions by modelling the arrival and departure times as variables and provide computational results for a scenario-based robust approach. In this paper, we aim to provide another model for the same railway network design problem under timetable constraints. A comparison of both models can be found in \cite{friesen2023modelling} where we describe the theoretical differences between the models and compare their computational results for a test instance of a network around Dresden.

The model presented here is based on a time expanded network. This considers the influence of timetables as well as the influence of uncertainty on the railway network design problem. We propose an optimization model that allows for the expansion cost-optimal calculation of railway networks while the demand is given as a set of trains with a respective starting and end node as well as earliest departure and latest arrival times. We prove that this problem is NP-hard to solve and present some graph classes for which the problem easier to solve. Additionally, we present a formalization of the optimization problem for the deterministic and the robust case.

This paper is structured as follows: First, we define the problem and introduce the notation in section \ref{notation}. In section \ref{model} we describe our mathematical model for the network design problem which we adapt for the robust approach in section \ref{robust}. Then, we provide the proof of the NP-hardness of the problem in \ref{complexity} as well as some easier special cases. Some computational results are given in section \ref{computational}. We conclude this paper in section \ref{conclusion} with some final remarks and a short outlook onto further research.

\section{Problem definition and notation}\label{notation}

In this paper, we expand the classic network design problem by incorporating timetable constraints and adapting railway specific capacity measures. We aim to find an expansion cost minimal railway network to operate a given long-term timetable on such that the capacity constraints are respected.


 The problem is closely related the well-known fixed-charge network design problem as described in \cite{thomadsen2007generalized}. Instead of costs for each unit of flow as in the more common minimum cost flow problem, we consider costs that occur if an arc is used at all, independently of the amount of flow on this arc. Here, these fixed charges correspond to the expansion costs for arcs. If the capacity of an arc is not sufficient, the arc is expanded and the expansion costs incur. We consider a setting in which it is a binary decision to expand an arc, so the arc is either expanded fully or not at all.

We consider a railway network consisting of stations represented by nodes and lines represented by arcs $(i,j)$. Each arc has a travelling time $tt_{ij}$ which is the number of time steps, a vehicle needs to pass that arc. 

Since the capacity of railway lines is limited, it is necessary to include a capacity measure to estimate the number of tracks needed for the operation of the input timetable. Railway line capacity depends on several different factors, most notably the type and amount of trains running on a line, the number of tracks and the train control system in place. Here, we assume the capacity to be given for each line. The capacity of each arc is given by $c_{ij}$ which is the maximum number of vehicles that can start in node $i$ to node $j$ in a predefined time interval $\widetilde{t}$. The capacity of a line can be extended by $\widetilde{c_{ij}}$ for the cost of $k_{ij}$. The lines that can be built, expanded or already exist are predefined as well as the corresponding costs. This is due to the fact that not all lines could be built due to financial and spacial restrictions. 

To integrate the temporal aspect we use a time-expanded network. That is, for each node $i$ we have as many copies $i_t$ as we have time steps $t\in T$ Where $T$ is the set of all time steps. The edges of the time expanded network rely on the edges of the original network. That is for each edge $ij$ in the original network we define edges $i_tj_{t'}$ in the time-expanded network where $t'-t$ is the travelling time for edge $ij$. To describe these edges we use a adjacency matrix $A$ where $a_{i_tj_{t'}}=1$ if the edge $ij$ exists and the travelling time between $i$ and $j$ is $t'-t$ and $t, t'\in T$. As we only model a fixed travel time for all vehicles in this paper, we omit $t'$ for simplicity in the following.

Further, we have a long-term timetable as an input which is described through a set of trains $V$ where each train $v\in V$ has a departure and an arrival node $d_v$ and $a_v$ with an earliest departure time $t_{d_v}$ and a latest arrival $t_{a_v}$ time at these nodes. These trains then have to be routed such that the, potentially expanded, capacity is not exceeded. 

For each edge and each pair of vehicles, there exists a minimum headway time $M_{v_1v_2ij}$ which states the minimal time train $v_2$ can leave onto line $ij$ after $v_1$ left station $i$. 

We define binary variables $b_{ij}$ that state if the edge $ij$ is expanded. The binary variables $x_{v,i,j,t}$ describe if the edge $ij\in E$ is used by vehicle $v\in V$ which leaves node $i\in N$ at time point $t\in T$. 

Additionally, we introduce VIA-nodes which are stations a train has to pass. For VIA-Nodes, we define a pair $(v,n)$ where train $v\in V$ needs to pass node $n\in N$ that is $\exists i\in N,t\in T: x_{i,n,v,t}=1$. Furthermore, connections are specified through a station and for the connection a pair of trains. In these stations, both trains have to stop and allow the transition of passengers from the first train to the second one.
Connections are given as a tuple $(n,v_1.v_2)$ where train $v_1$ should arrive at node $n$ before train $v_2$ to ensure a connection.

\section{Mixed Integer Linear Programming Formulation}\label{model}
In this section, we want to formalize the network design problem for railways under timetable constraints. We model this optimization problem through a time-expanded network.
The objective of the network design problem for railway infrastructure is to minimize the expansion costs of a railway network while the constrains ensure that the timetable is conductible. 

The optimization model can informally be described as follows:

\begin{align*}
& \underset{b_{ij}\text{ } binary, ij \in E}{\text{minimize}}
&\sum_{ij\in E} ExpansionCost_{i,j}\cdot b_{ij} \\
& &\\
&\text{such that}
&\text{(Expanded) capacity is not exceeded}  \\
& &\text{Earliest departure times are respected}\\
& &\text{Latest arrival times are respected}\\
& &\text{Minimum headway times are respected}\\
& &\text{Flows are conserved}\\
& & \text{Connections are respected}\\
& & \text{VIA-Nodes are respected}\\
\end{align*}
With the notations from Section \ref{notation} we obtain the following mathematical formalization.
\paragraph{Objective function}
\begin{align} \label{obj} 
\min \sum_{ij\in E} k_{i,j}\cdot b_{i,j} && \\ \nonumber
\end{align}
The objective function is to minimize the expansion costs. For that, the variable $b_{ij}$ for each arc ${ij}$ becomes true if the capacity of this arc is expanded. The costs for the expansion of that track $ij$ are given by $k_{ij}$. To obtain the costs for the whole network, we build the sum of these costs over all arcs.

\paragraph{Capacity}
\begin{align}
\label{cap}   
  \nonumber   \sum_{v\in V, t\in[t_0,t_0+\widetilde{t}]} x_{i,j,t,v}\leq c_{ij}+b_{ij}\cdot \widetilde{c_{ij}}& &\forall ij\in E, t_0\in T\\
  \end{align}

We include the capacity $c_{ij}$ as the number of trains on a line which cannot be exceeded per a given time interval $\widetilde{t}$. The binary variables $x_{i,j,v,t}$ are true if train $v$ leaves node $i$ at time $t$ to get to node $j$. Hence, the sum of these $x_{i,j,v,t}$ needs to be less or equal to the capacity $c_{ij}$ for a predefined time interval $\widetilde{t}\in\{1,\max\{T\}\}$. This needs to be respected for each time interval $[t_0,t_0+\widetilde{t}]$ for all time points $t_0\in \{0,\dots,\max\{T\}-\widetilde{t}\}$. This capacity $c_{ij}$ can be extended by $\widetilde{c_{ij}}$ to allow more trains per time window. If the line is extended, $b_{ij}$ has to be true and therefore the costs of the extension have to be respected in the objective function.

\paragraph{Departure and arrival}
\begin{align}
  \label{dep}
    \nonumber   \sum_{t<t_{d_v}} x_{d_v,j,t,v}=0 & & \forall v \in V, \forall j \in N\\
    \nonumber   \sum_{t\geq t_{d_v}} a_{d_{v_t},j_{t'}}\cdot x_{d_v,j,t,v}\geq 1 & & \forall v \in V, \forall j \in N\\
    \nonumber    \sum_{t>t_{a_v}} x_{i,a_v,t,v}=0 & &\forall v \in V, \forall i \in N\\
   \nonumber   \sum_{t\leq t_{a_v}} a_{i_{t'},a_{v_t}}\cdot x_{i,a_v,t-tt(i,a_v),v}\geq 1 & & \forall v \in V, \forall i \in N\\
      \end{align}
These constraints ensure that no vehicle $v$ leaves before their earliest departure time $t_{d_v}$ which is no $x_{i,j,t, v}$ is $1$ for $t<t_{d_v}$. Furthermore, they make sure that the departure from the starting node $d_v$ only occurs along an existing line. This is not ensured by the flow constraints as the starting and end nodes are not included in the flow constraints we use in this model. These constraints are implemented for the arrival node $a_v$ and latest arrival time $t_{a_v}$ accordingly.

\paragraph{Minimum Headway time}
\begin{align}
      \label{mht} 
  \nonumber  x_{i,j,t_1,v_1}\cdot (M_{ijv_1v_2}-(t_2-t_1))-(1-x_{i,j,t_2,v_2})\cdot max\{0;M_{ijv_1v_2}-(t_2-t_1)\}\leq 0 \\ \forall ij\in E,t_1<t_2\in T,v_1\neq v_2\in V
\end{align}
Furthermore, we implement minimum headway times. These are defined as the amount of time that has to pass between the departure of two trains following each other on the same track. They depend on the speed and the acceleration of the trains, the blocking distance and the distance to the next station which permits a change of the train sequence. Therefore, the minimum headway time is individually determined for each line $ij$ and each pair of trains $(v_1,v_2)$. 

The most intuitive way to implement the minimum headway time is to multiply the difference between the departure times and the minimum headway time with the binary variables $x_{i,j,t_1,v_1}$ and $x_{i,j,t_2,v_2}$ for all $ij\in E,t_1<t_2\in T,v_1\neq v_2\in V$ which indicate that the trains are driving on that line at the indicated time. By that, we obtain the constraint \begin{equation}
     x_{i,j,t_1,v_1}\cdot (M_{i,j,v_1,v_2}-(t_2-t_1))\cdot x_{i,j,t_2,v_2}\leq 0
\end{equation}
To avoid this quadratic constraint, we use the linear constraint given in \ref{mht}.
If both trains $v_1$ and $v_2$ use the line $(i,j)$, we obtain $(M_{ijv_1v_2}-(t_2-t_1))\leq 0$ which is exactly the constraint we are aiming for. If $x_{i,j,v_1,t_1}=1$ and $x_{i,j,t_2,v_2}=0$, we obtain $(M_{ijv_1v_2}-(t_2-t_1))-max\{0;(M_{ijv_1v_2}-(t_2-t_1))\}\leq 0$ which is true trivially.
If $x_{i,j,v_1,t_1}=0$, equation \ref{mht} becomes $(1-x_{i,j,t_2,v_2})\cdot -max\{0;M_{ijv_1v_2}-(t_2-t_1)\}\leq 0$ which is $0\leq 0$ for $x_{i,j,t_2,v_2}=1$ and $-max\{0;M_{ijv_1v_2}-(t_2-t_1)\}\leq 0$ for $x_{i,j,t_2,v_2}=0$. Both are easy to verify.
\paragraph{Flow constraints}
\begin{align}  \label{flow}   
  \nonumber   \sum_{i\in N,t\in T,j\neq a_v} a_{i_t,j_{t'}}\cdot x_{i,j,t,v}-\sum_{i\in N,t''\in T,j\neq d_v} a_{j_{t'},i_{t''}}\cdot x_{j,i,t',v} =0 \\\forall v\in V,j\in N,t'\in T   
\end{align}
The train paths have to be consistent. That is that no trains appear or disappear somewhere else than their origin or destination node. This is made sure through the inclusion of flow conservation constraints. All trains that enter a time-space-node have to leave that node through an existing line except for their arrival node.  
By including the parameter $a_{i_t,j_{t'}}$ which is true if the connection is possible we make sure that only existing lines are used.


  
\paragraph{Connections}
~\newline \noindent Connections are given as a tuple of two trains $v_1$ and $v_2$ and a station $i$ where the connection should occur. By the constraints
\begin{align}\label{connection}
      \sum_{n=1}^{|N|} a_{n_t,j_t'}\cdot x_{n,i,t-tt_{n,i},v_1}- \sum_{n=1}^{|N|} a_{i_t,n_{t'}}\cdot x_{i,n,t,v_2} \geq 0&& \forall t\in T, (v_1,v_2,i)\in Connections
\end{align}
we ensure that $v_1$ arrives before or at the same time as $v_2$. Furthermore, we need to ensure that train $v_2$ leaves station $i$. This is done by
\begin{align}
   \sum_{t\in T}  \sum_{n=1}^{|N|} x_{i,n,t,v_2} \geq 1&& \forall (v_1,v_2,i)\in Connections
\end{align}
This already implies by equation \ref{connection} the arrival of train $v_1$ at the station as elsewise equation (\ref{connection}) could become negative.
\paragraph{VIA-Nodes}
\begin{align}\label{VIA}
     \sum_{t\in T}  \sum_{j=1}^{|N|} x_{i,j,t,v} \geq 1 && \forall (v,i)\in VIA-Nodes
\end{align}
We can define a VIA-Node $i$ for a train $v$ as a tuple $(v,i)\in VIA-Nodes$. This ensures that the vehicle $v$ passes node $i$ at some time point.

\subsection{Robust Extension of the Model} \label{robust}
Since extension measures have a long planning and construction horizon, the long-term time\-table, we use to determine the necessary extensions, is not final. The demand or the traffic policy in the future could differ. To incorporate this uncertainty of the timetable, we developed a robust version of the model.

We model the uncertainty through two different extensions of the deterministic model:
\begin{itemize}
    \item Optional trains and
    \item Different timetables as scenarios. 
\end{itemize}

First, we introduce optional trains, which can but do not have to be routed. If they are not routed, this leads to a penalty which is implemented in the objective function. The objective function is therefore changed to  
\begin{equation}
\min \sum_{ij \in E} k_{i,j}\cdot b_{i,j}-\sum_{v\in V_{opt}} x_{d_v,j,t,v}\cdot k_v
\end{equation}
where $k_v$ is the penalty of the optional train $v$ and $V_{opt}\subseteq V$ is the set of optional trains. The departure and arrival are not necessary for the optional trains. Therefore, the constraints that force the train to leave and arrive are omitted for the optional trains $v\in V_{opt}$. 
\begin{align}
 \label{dep2}
    \text{Arrival and Departure} & \\ 
    \nonumber   \sum_{t<t_{d_v}} x_{d_v,j,t,v}=0 & & \forall v \in V \\
    \nonumber   \sum_{t\geq t_{d_v}} a_{d_{v_t},j_{t'}}\cdot x_{d_v,j,t,v}\geq 1 & & \forall v \in V\backslash V_{opt}\\
    \nonumber    \sum_{t>t_{a_v}} x_{i,a_v,t,v}=0 & &\forall v \in V\end{align}\begin{align}
   \nonumber   \sum_{t\leq t_{a_v}} a_{i_{t'},a_{v_t}}\cdot x_{i,a_v,t-tt(i,a_v),v}\geq 1 & & \forall v \in V\backslash V_{opt}
   \end{align}

   As all $x_{i,j,v,t}$ on the time-space-path of the train equal $x_{d_v,j,v,t}$ for some $t$ due to the flow constraints, the path is completely true or all variables for this train are $0$. Therefore, the capacity constraints and minimum headway time constraints do not change due to the introduction of optional trains.  Connection constraints and VIA-node constraints are not implemented for optional trains.
   
Secondly, we integrated $m$ different scenarios into the optimization. Each scenario has its own set of trains $V_n$ where each train in $V=\bigcup_{n\in\{1,\dots,m\}} V_n$ has to be routed. Therefore the departure and arrival time constraints (\ref{dep}), the flow constraints (\ref{flow}), the connection constraints (\ref{connection}) and VIA-node constraints (\ref{VIA}) stay the same and have to be fulfilled for all trains. The constraints (\ref{mht2}) and (\ref{cap2}) describe interactions between the trains and are used for every $V_n$ independently:

\begin{align}
  \label{cap2}   \text{Capacity}\\ 
  \nonumber\sum_{v\in V_n, t\in(t_0,t_0+\widetilde{t})} x_{i,j,t,v}\leq c_{ij}+b_{ij}\cdot \widetilde{c_{ij}}\\\nonumber\forall ij\in E, t_0\in T,t_0<\max\{T\}-\widetilde{t}, n\in \{1,\dots,m\}\end{align}
  \begin{align}\label{mht2}  \text{Minimum Headway times}\\
  \nonumber  x_{i,j,t_1,v_1}\cdot (M_{ijv_1v_2}-(t_2-t_1))-(1-x_{i,j,t_2,v_2})\cdot max\{0;M_{ijv_1v_2}-(t_2-t_1)\}\leq 0 \\\nonumber\forall ij\in E,t_1<t_2\in T,v_1\neq v_2\in V_n , n\in \{1,\dots,m\}
    \end{align}



With these adaptions, we can model a uncertainty set of timetables in which different scenarios with different set of trains can occur. 
The optimization model can decide if the network can be extended to be able to schedule optional trains. This happens through the comparison of the penalty for not scheduling the train and the needed expansion cost for scheduling it. This uncertainty set allows us to represent a uncertain long-term timetable for operation in a few decades. Based on this, we can decide which infrastructure measure can be taken to obtain a robust infrastructure.

\subsection{Computational Results}\label{computational}
The following test instances are run on an Intel i7-10700 CPU 2.90GHz machine with 16 GB
RAM using Gurobi 9.1.0 as optimizer. All the instances are solved
to optimality. The runtime is taken as the average out of 4 resolutions of the same instance. The test instances are a small extract of one hour of the German strategic timetable \"Deutschlandtakt\" in the area of Dresden, see \cite{SMAundPartnerAG.2020.ITF-SO}.

In the deterministic case, we obtain the running times in table 1:

\begin{table}[H] \label{det}
    \centering
    \begin{tabular}{|l|l|l|l|l|l|l|l|}
    \hline
        Timesteps & Trains & Nodes & Arcs & Runtime [s] & Constraints & Variables \\ \hline
       60	&20	&30	&100	&0,7	&19875	&1661  \\ \hline
       60	&22	&30&	100&	1,73	&26621&	2313
\\ \hline
        60	&24	&30&	100&	1,81	&32265	&2965
 \\ \hline
      60	&26	&30	&100	&3,77&	44879	&4019
 \\ \hline
     60	&28&	30	&100	&16,5	&56507&	5073
 \\ \hline
       
    \end{tabular}
        \caption{Runtime of the optimization model for the deterministic case}
\end{table}

Already in these small examples, we see the exponentially rising runtime and number of constraints. 

For the case with more scenarios, the runtime heavily depends on the partition of the vehicles onto the scenarios as shown in table 2. If the vehicles are approximately evenly distributed over the scenarios, the runtime decreases compared to the deterministic case. This happens as the constraints for the minimum headway time (which are the majority of the constraints) only occur for vehicles that are in the same scenario. If a higher percentage of the vehicles occur in both scenarios, the number of constraints and the runtime rise.

\begin{table}[H]\label{szenario}
    \centering
    \begin{tabular}{|l|l|l|l|l|l|l|l|l|l|l|}
    \hline
       Trains & in Sc 1 & in Sc 2 &in Sc 3 & in Sc 4& Nodes & Arcs &  Runtime[s]& Constr. & Variables \\ \hline
        28 & 7 & 7 & 7 & 7 & 30 &100 &1,19 & 45502 & 5076 \\ \hline
         28 & 9 & 8 & 7 & 2 & 30 & 100& 1,59 & 45284 & 5076\\ \hline
         28 & 10 & 9 & 7 & 1 & 30 & 100& 1,15 & 45612 & 5076\\ \hline
       28 & 12 & 10 & 4 & 0 & 30 & 100& 0,54 & 45612 & 5076\\ \hline
    \end{tabular}
    \caption{Runtime of the optimization model for four scenarios and 60 timesteps}
\end{table}

\section{Complexity}\label{complexity}
In this section, we are going to analyze the complexity of the network design problem for railways under timetable constraints. We will show, that the problem is already NP-hard for very simple graphs and timetable structures. Nevertheless, we will give some examples of specific graph classes in which the problem is (pseudo-)polynomial solvable. We will focus on given, fixed departure times. Otherwise, we obtain, even if our network consists of only one edge, an interval scheduling problem. This is shown to be NP-hard in \cite{chuzhoy2006approximation} if there exists three or more departure options for any train.

The proof in this section is adapted from \cite{holzhauser2016budget} who studied the budget-constrained minimum cost flow problem. The network design problem for railways we study in here is analogous to the dual problem of the budget-constrained minimum cost flow problem. In this problem, a minimum cost flow is searched for. The costs occur per unit of flow on each arc. Furthermore, a fixed-charge is applied if the capacity of an arc is not sufficient. This fixed charge is budgeted. 

For this proof, we reduce the ExactCoverBy3Sets problem onto the railway network design problem. In this proof, we omit VIA-nodes and connections for simplicity.

\begin{definition}ExactCoverBy3Sets: Given a set $X$ with $3q$ elements and a set $C$ of
3-element subsets $C_i$ of $X$. Does there exist a subset $\widetilde{C}\subset C$ where
each element of $X$ occurs exactly once?
\end{definition}
\begin{theorem}
The railway network design problem is strongly NP-hard to solve even on bipartite graphs and if all trains have the same, fixed departure time.
\end{theorem}
\begin{proof} 
We reduce the ExactCoverBy3Sets-problem to our problem. This problem is shown to be strongly NP-hard in \cite{garey1979computers}.

If we are given an instance of ExactCoverBy3Sets, we design a network consisting of a source and a sink vertex, a vertex $v'_i$ for each subset $C_i$ and a vertex $v_j$ for each element $x_j \in X$. We then construct arcs from the source $s$ to each node $v'_i$ and from each node $v_i$ to the sink $t$. Furthermore, we connect each node $v'_i$ with the nodes $v_j$ if $x_j\in C_i$. All of these arcs have a capacity of 0. Let the expandable capacity $\widetilde{c_{s,v'_i}}$ of the arcs between the source and each $v'_i$ be $3$ and the expandable capacity of the arcs between each $v_i$ and $t$ and the expandable capacity of the arcs between $v'_i$ and $v_j$ equal 1. Let all arcs have a travelling time of 1. The expansion cost for the arcs from the source to each $v'_i$ are 3 while all other arcs have expansion costs of 0. Alternatively, this can be modelled by defining no expandable capacity and having already a capacity of 1 in these arcs.

For a small example given by $X=\{x_1,\dots,x_6\}$ and \newline $C=\big\{\{x_1,x_3,x_4\},\{x_1,x_4,x_5\},\{x_2,x_5,x_6\}\big\}$, the resulting network is shown in figure 1.
\newline~\newline
\begin{figure}
    \centering
    \includegraphics[width=0.8\textwidth]{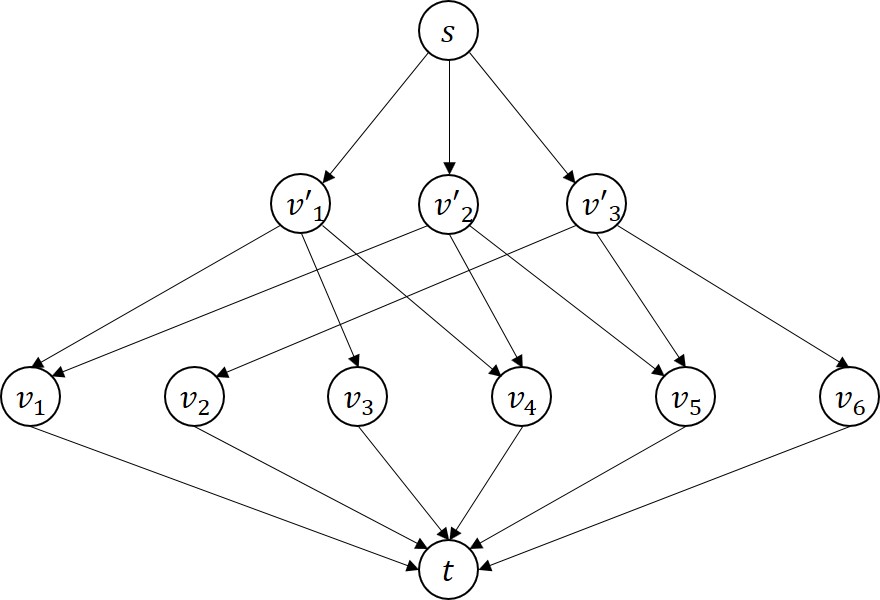}\label{cover}
 \caption{Resulting graph modelling an Exact3Cover as a railway network design problem}
\end{figure}

In this example we can easily verify that the sets corresponding to $v'_1$ and $v_3'$ would lead to an Exact3Cover.

We define a timetable with $X$ trains which all have the same earliest departure time $0$ at their departure node $s$ and latest arrival time $3$ at their arrival node $t$. This results in a given, specific departure time which is the same for all trains (and as the minimal travel time is $3$, the departure has to be at the earliest possible departure time). All trains start at node $s$ and end at node $t$.

We consider the decision problem if it is possible to construct a routing with a maximum cost of $3q$ such that $3q$ trains can be operated. We show that this is true if and only if the given instance of ExactCoverBy3Sets is a yes-instance.

Let there exist a schedule for $3q$ trains with cost $\leq 3q$. In this case at most $q$ lines from the source to the nodes $v'_i$ are expanded and can therefore be used to route the trains. Since the capacity of the arcs between $s$ and $v'_i$ is 3 each and $3q$ trains are scheduled, we obtain that at least $q$ arcs have to be expanded. Therefore, 3 trains have to use each of the expanded arcs between the source and the $v'_i$ since there exists exactly as many options to leave $s$ as there exists trains. Each trains then has to use exactly one of the edges between $v'_i$ and $v_j$. Since only $q$ nodes $v'_i$ are used, all arcs between these $v'_i$ and their corresponding three successors $v_j$ have to be used to route the trains. This leads to a subset $\widetilde{C}=\{C_i\vert (s,v'_i) \text{ is expanded}\}$ with $q$ sets such that all elements of $X$ are covered.

Let us now assume that there exists an ExactCoverBy3Sets given as $\widetilde{C}\subset C$. With the same correspondence as before, we can expand the arcs $(s, v'_i)$ for each $C_i\in \widetilde{C}$. This leads to expansion costs of $3q$. Then three trains can be send through each of these arcs. As there exists an exact cover, these trains can then proceed to the three nodes $v_j$ following $v'_i$ and each $v_j$ is only visited by one train for which the capacity of 1 between $v_j$ and $t$ is sufficient. This leads to a routing for all $3q$ trains with a cost of $3q$. By this, we show that ExactCoverBy3Sets can be reduced polynomially onto the network design problem for railways. If we have a train schedule, we can easily verify if it is feasible. Therefore, the problem is in NP and we have shown that it is strongly NP-hard. 
\end{proof}

As we used an extremely simple timetable and a svery simple costs structure to prove the NP-hardness, a further simplification of the timetable does not seem helpful to obtain easier solvable special cases. Therefore, we focus in the following on specific graph classes to find (pseudo-)polynomial solvable subproblems. 

\subsection{Complexity on Arborescences}
In this subsection, we show that the network design problem for railway is polynomially solvable on arborescences. 
An arborescences is an abstraction of a tree onto digraphs. One vertex is defined as a root and all arcs point away from the root. There exists no cycles. This leads to a unique path between each pair of nodes.

\begin{theorem}
The railway network design problem is polynomially solvable on an arborescence, if all trains leave at a fixed departure time.
\end{theorem}
\begin{proof}
If we are given a set of trains with an origin and a departure, the route for each train is unique and therefore fixed. Hence, we can count the number of vehicles on each arc (either per time step or directly if we set the root vertex as departure for all trains). We can now verify if the routing is feasible and if the (expandable) capacity on each arc is respected. It follows directly which arcs need to be expanded. It can be easily seen that all these steps can be executed in polynomial time in the input size.
\end{proof}

However, as the trains potentially use the same lines, we again obtain an interval scheduling problem if we omit the fixed departure time.

\subsection{Complexity on Series-parallel graphs}
We focus in this section on so called series-parallel graphs, for which we obtain pseudo-polynomial solvability. We neglect here the fact that trains use the capacity only during their respective travelling time on the line as we would again obtain an NP-hard interval scheduling problem even on one arc. (This corresponds to choosing the time interval for the capacity constraints $\widetilde{t}=\max{t}$.)

\begin{definition}
A series-parallel graph is a graph $G=(V,E)$ with a source $s$ and a sink $t$ which can be recursively constructed. Each series-parallel graph is build from two series-parallel graphs through composition as shown in figure 2. The simplest series-parallel graph is defined as one edge between $s$ and $t$. This simple graph can then be expanded through series and parallel composition. For a series composition of two series-parallel graphs $G$ and $G'$ the sink $t$ of $G$ is contracted with the source $s'$ of $G'$ to obtain a new graph with source $s$ an sink $t'$. For a parallel composition the sources $s$ and $s'$ are contracted to obtain a new source $s$ as well as the sinks $t$ and $t'$ are contracted to a new sink $t$.
\end{definition}  
\begin{figure}
    \centering
    \includegraphics[width=0.7\textwidth]{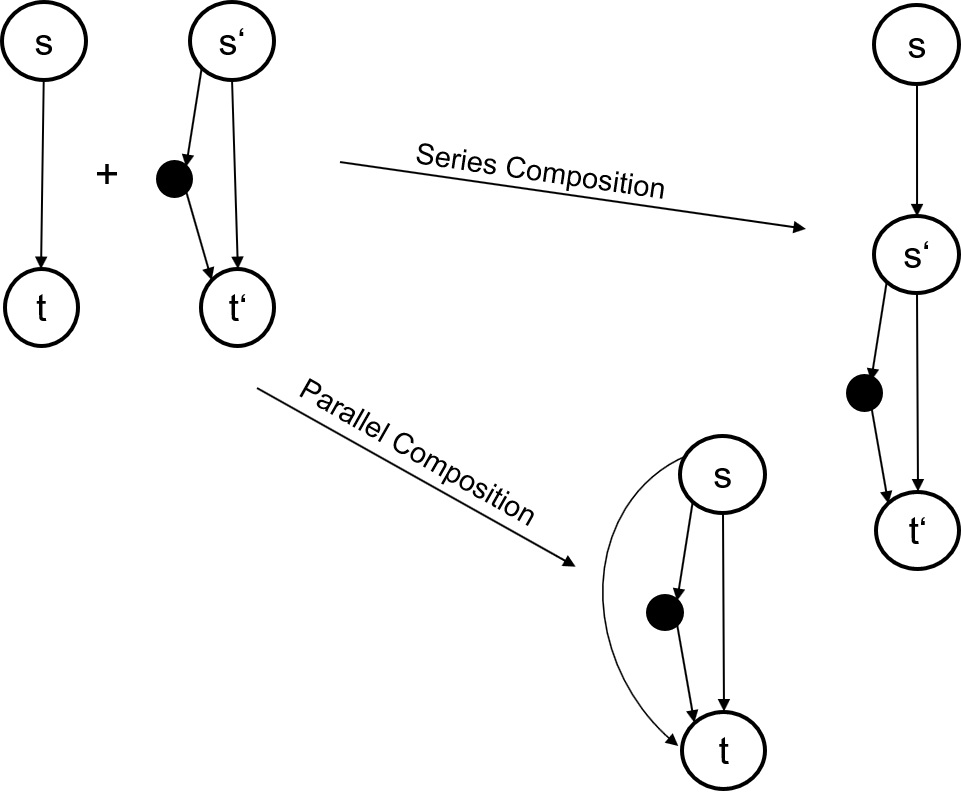}
    \caption{Series and parallel composition for a series-parallel graph}
    \label{fig:SPGraph}
\end{figure}

\begin{theorem}
The network design problem for railway infrastructure under timetable constraints is pseudo-polynomial solvable in 
\begin{equation}
 O(\vert V\vert \cdot max_v(tt_v)\cdot (1\cdot m+ max_v(tt_v)\cdot n+\vert V\vert \cdot  m)
\end{equation}
for problems with a uniform given departure time.
\end{theorem}

\begin{proof}
By $K_G(v,tt)$ we denote the minimal expansion cost on a series-parallel graph for a number of $v$ vehicles and a maximum travel time for all vehicle of $tt$ time steps.

We start with the case that the network $G=(N,A)$ only consists of one edge. 
If the extendable capacity is smaller than the number of trains that want to pass the edge, there exists no feasible solution and the minimum costs are infinite. If the capacity without any extension is already greater or equal than the number of trains, no extensions are necessary, so no costs occur. Finally, if the extended capacity is greater or equal to the number of trains while the existing capacity is smaller, the minimal costs are exactly the extension costs for this edge.

Therefore, we obtain for the expansion costs for the edge $ij$
\begin{equation}
  K_G(v, tt) = \begin{cases}
0 &\text{if } v\leq c_{ij}\\
k_{ij} &\text{if }c_{ij}<v\leq c_{ij}+\widetilde{c}_{ij}\\
\infty &\text{if }v\geq c_{ij}+\widetilde{c}_{ij}
\end{cases}
\end{equation}

\paragraph{Series Composition}
For a series composition $G=G_1\circ G_2$ we show that 
\begin{equation}
K_G(v,tt)=\min_{tt_1\in (0,tt)} K_{G_1}(v,tt_1)+ K_{G_2}(v,tt-tt_1).
\end{equation}

The minimal expansion costs for graph $G$ denoted by $K_G(v,tt)$ are clearly smaller or equal to the sum of the costs $K_{G_1}(v,tt_1)+ K_{G_2}(v,tt-tt_1)$ for al $tt_1\in (0,tt)$ as this is a specific division of $G$ into subgraphs for which the expansion scheme is minimal. This is not necessarily the minimal expansion scheme for $G$.

Let us assume there exists a flow $x$ such that 
\begin{equation}
K_G(v,tt)<\min_{ tt_1\in (0,tt)} K_{G_1}(v,tt_1)+ K_{G_2}(v,tt-tt_1).
\end{equation}
Let $x_1$ denote the flow of $G$ restricted on $G_1$ and $x_2$ the flow restricted on $G_2$. Then either $x_1$ has to have expansion costs lower than $K_{G_1}(v,tt_1)$ or $x_2$ has to have expansion costs lower than $K_{G_2}(v,tt-tt_1)$ for some $tt_1\in (0,tt)$. This contradicts the definition of $K_{G_1}$ as minimal, as the flow $x$ has to pass $G_1$ completely and then pass $G_2$ completely by the definition of a series-parallel graph.

\paragraph{Parallel Composition}
Let $G=G_1\vert G_2$ be the parallel composition of two series-parallel graphs $G_1$ and $G_2$. 
Then we obtain for the expansion cost of $G$\begin{equation}
    K_G(v,tt)=\min_{v_1\leq v} K_{G_1}(v_1,tt)+ K_{G_2}(v-v_1,tt)
\end{equation}
This holds, since we look for the optimal partition of the vehicles onto $G_1$ and $G_2$ with respect to the expansion costs.

For each series-parallel graph, a tree can be constructed whose nodes refer to the series and parallel compositions and whose leaves represent each single edge.
In \cite{valdes1982recognition} Valdes et al. show that such a tree for $G$ can be computed in $O(m+n)$ where $m$ is the number of arcs and $n$ the number of vertices in $G$. This tree has $O(m)$ inner nodes for parallel composition and $O(n)$ inner nodes for series composition as well as $m$ leaves which correspond to single arcs. We compute $K_G(v,tt)$ for each node in the composition tree for $tt\leq \max_v(tt_v)$ and $v\leq \vert V\vert$. If the values are computed from the leaves up, it can be assumed that the values for all subgraphs are known. Therefore, the value $K_G$ can be computed in $O(1)$ for a single edge, in $O(max_v(tt_v))$ for a series composition as we minimize over the travelling times and in $O(V)$ for a parallel composition as we minimize over the partition of the vehicles onto the parallel graphs. Then we obtain the minimal expansion cost with $K_G(\vert V\vert,  max_v(tt_v))$. Therefore, we obtain a running time of
\begin{equation}
 O(\vert V\vert \cdot max_v(tt_v)\cdot (1\cdot m+ max_v(tt_v)\cdot n+\vert V\vert \cdot  m) .
\end{equation}
\end{proof}

\section{Conclusion and Outlook}\label{conclusion}
In this article, we analysed the network design problem under timetable constraints.
We have shown that the railway network design problem under timetable constraints is NP-hard to solve on bipartite graphs and  have shown some easier solvable cases. Furthermore, we provided a formalization of the macroscopic railway network design problem under timetable constraints for both the deterministic and the robust problem formulation using time expanded networks. The robust formulation can be used as a long-term timetable. This still allows some flexibility to adapt the timetable to changing traffic demand or political conditions during the long planning process of railway infrastructure. Based on this timetable, the infrastructure planning can be conducted more demand-oriented and economical.

Our next research goal is to find heuristics and approximations to reduce the running time of the implementation of this NP-hard optimization problem.
\paragraph{Acknowledgements}
This work is funded by the Deutsche Forschungsgemeinschaft
(DFG, German Research Foundation) – 2236/1 and NI 1597/4-1.

%

\bibliographystyle{plain} 
\bibliography{bib}

\begin{thebibliography}{10}

\bibitem{chen2011transport}
Anthony Chen, Zhong Zhou, Piya Chootinan, Seungkyu Ryu, Chao Yang, and SC~Wong.
\newblock Transport network design problem under uncertainty: a review and new
  developments.
\newblock {\em Transport Reviews}, 31(6):743--768, 2011.

\bibitem{chuzhoy2006approximation}
Julia Chuzhoy, Rafail Ostrovsky, and Yuval Rabani.
\newblock Approximation algorithms for the job interval selection problem and
  related scheduling problems.
\newblock {\em Mathematics of Operations Research}, 31(4):730--738, 2006.

\bibitem{costa2009benders}
Alysson~M Costa, Jean-Fran{\c{c}}ois Cordeau, and Bernard Gendron.
\newblock Benders, metric and cutset inequalities for multicommodity
  capacitated network design.
\newblock {\em Computational Optimization and Applications}, 42(3):371--392,
  2009.

\bibitem{friesen2023modelling}
Nadine Friesen, Tim Sander, Karl Nachtigall, and Nils Nie{\ss}en.
\newblock Modelling time in the timetable-based railway network design problem.
\newblock {\em Available at SSRN 4470191}.

\bibitem{garcia2013grasp}
Bosco Garc{\'\i}a-Archilla, Antonio~J Lozano, Juan~A Mesa, and Federico Perea.
\newblock Grasp algorithms for the robust railway network design problem.
\newblock {\em Journal of Heuristics}, 19(2):399--422, 2013.

\bibitem{garey1979computers}
Michael~R Garey and David~S Johnson.
\newblock Computers and intractability.
\newblock {\em A Guide to the}, 1979.

\bibitem{gendron1999multicommodity}
Bernard Gendron, Teodor~Gabriel Crainic, and Antonio Frangioni.
\newblock Multicommodity capacitated network design.
\newblock In {\em Telecommunications network planning}, pages 1--19. Springer,
  1999.

\bibitem{grujivcic2015variable}
Igor Gruji{\v{c}}i{\'c}, G{\"u}nther Raidl, and Andreas Sch{\"o}bel.
\newblock Variable neighborhood search for integrated timetable based design of
  railway infrastructure.
\newblock {\em Electronic Notes in Discrete Mathematics}, 47:141--148, 2015.

\bibitem{grujivcic2017metaheuristic}
Igor Gruji{\v{c}}i{\'c}, G{\"u}nther Raidl, Andreas Sch{\"o}bel, and Gerhard
  Besau.
\newblock A metaheuristic approach for integrated timetable based design of
  railway infrastructure.
\newblock {\em Road and Rail Infrastructure III}, 2017.

\bibitem{guihaire2008transit}
Val{\'e}rie Guihaire and Jin-Kao Hao.
\newblock Transit network design and scheduling: A global review.
\newblock {\em Transportation Research Part A: Policy and Practice},
  42(10):1251--1273, 2008.

\bibitem{heppe2019spurplangestaltung}
Andreas Heppe and Werner Weigand.
\newblock Spurplangestaltung und betriebliche infrastrukturplanung.
\newblock In {\em Handbuch Eisenbahninfrastruktur}, pages 467--523. Springer,
  2019.

\bibitem{holzhauser2016budget}
Michael Holzhauser, Sven~O Krumke, and Clemens Thielen.
\newblock Budget-constrained minimum cost flows.
\newblock {\em Journal of Combinatorial Optimization}, 31(4):1720--1745, 2016.

\bibitem{kuby2001railway}
Michael Kuby, Zhongyi Xu, and Xiaodong Xie.
\newblock Railway network design with multiple project stages and time
  sequencing.
\newblock {\em Journal of Geographical Systems}, 3(1):25--47, 2001.

\bibitem{leblanc1975algorithm}
Larry~J Leblanc.
\newblock An algorithm for the discrete network design problem.
\newblock {\em Transportation Science}, 9(3):183--199, 1975.

\bibitem{lusby2018survey}
Richard~M Lusby, Jesper Larsen, and Simon Bull.
\newblock A survey on robustness in railway planning.
\newblock {\em European Journal of Operational Research}, 266(1):1--15, 2018.

\bibitem{magnanti1984network}
Thomas~L Magnanti and Richard~T Wong.
\newblock Network design and transportation planning: Models and algorithms.
\newblock {\em Transportation science}, 18(1):1--55, 1984.

\bibitem{mudchanatongsuk2008robust}
Supakorn Mudchanatongsuk, Fernando Ord{\'o}{\~n}ez, and Jie Liu.
\newblock Robust solutions for network design under transportation cost and
  demand uncertainty.
\newblock {\em Journal of the Operational Research Society}, 59(5):652--662,
  2008.

\bibitem{pishvaee2012robust}
Mir~Saman Pishvaee, Jafar Razmi, and S~Ali Torabi.
\newblock Robust possibilistic programming for socially responsible supply
  chain network design: A new approach.
\newblock {\em Fuzzy sets and systems}, 206:1--20, 2012.

\bibitem{polinder2021timetabling}
Gert-Jaap Polinder, Marie Schmidt, and Dennis Huisman.
\newblock Timetabling for strategic passenger railway planning.
\newblock {\em Transportation Research Part B: Methodological}, 146:111--135,
  2021.

\bibitem{sander2023robust}
Tim Sander, Nadine Friesen, Karl Nachtigall, and Nils Nie{ß}en.
\newblock Robust railway network design based on strategic timetables.
\newblock submitted.

\bibitem{schobel2012timetable}
Andreas Sch{\"o}bel and Gerhard Besau.
\newblock Timetable based design of railway infrastructure.
\newblock In {\em EURO-ZEL 20th International Symposium}, 2012.

\bibitem{schobel2013optimization}
Andreas Sch{\"o}bel, GR~Raidl, Igor Grujicic, Gerhard Besau, and Gottfried
  Schuster.
\newblock An optimization model for integrated timetable based design of
  railway infrastructure.
\newblock In {\em Proceedings of the 5th International Seminar on Railway
  Operations Modelling and Analysis--RailCopenhagen}, pages 765--774, 2013.

\bibitem{SMAundPartnerAG.2020.ITF-SO}
{SMA und Partner AG}.
\newblock Zielfahrplan {D}eutschlandtakt: {D}ritter {G}utachterentwurf
  {S}achsen/{S}achsen-{A}nhalt/{T}h{\"u}ringen, 2020.

\bibitem{sponemann2013network}
Jacob~Christian Sp{\"o}nemann.
\newblock {\em Network design for railway infrastructure by means of linear
  programming}.
\newblock PhD thesis, Aachen, Techn. Hochsch., Diss., 2013, 2013.

\bibitem{thomadsen2007generalized}
Tommy Thomadsen and Thomas Stidsen.
\newblock The generalized fixed-charge network design problem.
\newblock {\em Computers \& Operations Research}, 34(4):997--1007, 2007.

\bibitem{ukkusuri2007robust}
Satish~V Ukkusuri, Tom~V Mathew, and S~Travis Waller.
\newblock Robust transportation network design under demand uncertainty.
\newblock {\em Computer-Aided Civil and Infrastructure Engineering},
  22(1):6--18, 2007.

\bibitem{valdes1982recognition}
Jacobo Valdes, Robert~E Tarjan, and Eugene~L Lawler.
\newblock The recognition of series parallel digraphs.
\newblock {\em SIAM Journal on Computing}, 11(2):298--313, 1982.

\bibitem{zhao2008optimization}
Fang Zhao and Xiaogang Zeng.
\newblock Optimization of transit route network, vehicle headways and
  timetables for large-scale transit networks.
\newblock {\em European Journal of Operational Research}, 186(2):841--855,
  2008.

\end{thebibliography}


\end{document}